\documentclass[12pt,onecolumn,journal,letterpaper]{IEEEtran}

\usepackage{ae,aecompl}
 \usepackage{color}
\usepackage{amsmath}
\usepackage{theorem}
\usepackage{dsfont}
\usepackage{amsfonts}
\usepackage{amssymb}
\usepackage{cite}
\usepackage{cancel}

\newtheorem{thm}{Theorem}

\date{January 2007}

\begin{document}

\title{On the Degrees of Freedom of the Compound MIMO Broadcast Channels with Finite States
\thanks{Part of the materials of this paper has been independently reported in~\cite{Jafar_compund}.}}

\author{Mohammad Ali Maddah-Ali,\\
Department of Electrical and Computer Engineering,\\
University of California - Berkeley
}

 \maketitle

\begin{abstract}
Multiple-antenna broadcast channels with $M$ transmit antennas and $K$ single-antenna receivers is considered, where the channel of receiver $r$ takes one of the  $J_r$ finite values. It is assumed that the channel states of each receiver are randomly selected from $\mathds{R}^{M\times 1}$ (or from $\mathds{C}^{M\times 1}$). It is shown that no matter what $J_r$ is, the degrees of freedom (DoF) of $\frac{MK}{M+K-1}$ is achievable. The achievable scheme relies on the idea of interference alignment at receivers, without exploiting the possibility of cooperation among transmit antennas. It is proven that if $J_r \geq M$, $r=1,\ldots,K$, this scheme achieves the optimal DoF. This results implies that when the uncertainty of the base station about the channel realization is considerable, the system loses the gain of cooperation. However, it still benefits from the gain of interference alignment. In fact, in this case, the compound broadcast channel is treated as a compound X channel.

Moreover, it is shown that when the base station knows the channel states of some of the receivers, a combination of transmit cooperation and interference alignment would achieve the optimal DoF.

Like time-invariant $K$-user interference channels,  the naive vector-space approaches of interference management seem insufficient to achieve the optimal DoF of this channel. In this paper, we use the Number-Theory approach of alignment, recently developed by Motahari et al.~\cite{abolfazl_k_2}. We extend the approach of~\cite{abolfazl_k_2} to complex channels as well, therefore all the results that we present are valid for both real and complex channels.
\end{abstract}

\begin{keywords}
Compound Broadcast Channel, Compound X Channel, Interference Alignment, Number Theory, Diophantine Approximation.
\end{keywords}

\section{Introduction}
We consider a real (or complex) compound broadcast channel with  $M$ transmit antennas and $K$ receivers, each equipped with a single antenna. The channel of the receiver $r$ takes one of the $J_r$ finite values. This channel can be modeled as
\begin{eqnarray}\label{eq:channel_model}
y_r^{\{s\}}[m]=\mathbf{h}_r^{{\{s\}}^{\dagger}} \mathbf{x}[m]+z_r^{\{s\}}[m], \quad r=1,\ldots,K, \quad s=1,\ldots,J_r,
\end{eqnarray}
where $\mathbf{x}[m] \in \mathds{R}^M (or \mathds{C}^M) $,  $\mathbb{E}(\mathbf{x}^{\dagger} [m] \mathbf{x}[m])\leq P$,  $z_s[m] \sim \mathcal{N}(0,1)$ (or $~\sim \mathcal{CN}(0,1)$), and the sequences $z_r[m]$'s are i.i.d. and mutually independent. In addition, $\mathbf{h}_r^{\{s\}} = [{h_{r1}^{\{s\}}}^{\dagger},\ldots, {h_{rM}^{\{s\}}}^{\dagger}]^{\dagger} \in \mathds{R}^M$ (or $\mathds{C}^M$) . We assume that the channel state of each  receiver is perfectly known at the user, but not at the transmitter. However, the transmitter is aware of the set of all possible channel realizations.

An important measure to approximate the capacity of a wireless channel is known as degrees of freedom (DoF). The DoF of a channel shows how the capacity of the real channel is scaled with $\frac{1}{2}\log_2 P$ or how the capacity of the complex channel is scaled with $\log_2 P$ , for large transmit power $P$. Formally, the optimal real DoF of a real wireless channel, is given by 
\begin{eqnarray}
d_{real}= \lim_{P \rightarrow \infty} \frac{C_{sum}}{\frac{1}{2}\log_2(P)},
\end{eqnarray}
and the optimal complex DoF of a complex wireless channel is given by
\begin{eqnarray}
d_{complex}= \lim_{P \rightarrow \infty} \frac{C_{sum}}{\log_2(P)},
\end{eqnarray}
where $C_{sum}$ is the sum-capacity of the channel. In fact, DoF gives a first-order approximation of the sum-capacity as $C_{sum} = \frac{d}{2}\log_2(P)+ o(P)$ for real channels and  $C_{sum} = d\log_2(P)+ o(P)$ for complex ones.

We note that when a stream of data is transmitted to one receiver, it causes interference over the other receivers. The problem of characterizing the optimal DoF of a channel is essentially equivalent to developing the most efficient way of interference management.
%

For the MIMO broadcast channel, the most well-know approach for interference management is known as zero-forcing. In this approach, the base station uses a channel-inversion precoder to force the cross-interference components to be zero and guarantee interference-free receivers.  Using this technique, we can achieve the optimal DoF of the channel~\eqref{eq:channel_model}, if $J_r=1$ for $r=1,\ldots,K$. In this case, the capacity of the channel is fully characterized~\cite{Br_Cap_RG_J}  and the DoF of the system is proven to be  $\min\{K,M\}$.  In~\cite{Hannan_BCCM}, it is shown that if the zero-forcing precoder is expanded from space dimension to space/time dimensions,  it achieves the optimal DoF of some compound MIMO  broadcast channels~\eqref{eq:channel_model}. More precisely, it is shown that the space/time zero-forcing approach achieves the optimal DoF of $1+\frac{M-1}{M}$ for $K=2$, $J_1=1$, and $J_2=M$, and achieves the optimal DoF of $\frac{2M}{M+1}$ for $K=2$ and $J_1=J_2=M$ ~\cite{Hannan_BCCM}. Moreover, it is shown that when $K=2$, $J_1=J_2=J \geq M$, this scheme yields the DoF of $\frac{2J}{2J-M+1}$. However, in terms of converse, it is proven that the DoF of the later channel is upper-bounded by $\frac{2M}{M+1}$. Then, it is conjectured that the gap between inner and outer bounds is duo to the looseness of the outer-bound, and the achievable scheme is optimal. In other words, it is believed that as $J$ increases, the optimal DoF converges to one. In another effort, in~\cite{Lapidoth}, it is shown that in ergodic MIMO broadcast channels, with $M=2$ and $K=2$, where the channel state information is not known at the transmitter, the DoF is upper-bounded by $\frac{4}{3}$. This is the same as the outer-bound of~\cite{Hannan_BCCM} for the corresponding compound channel.

In the context of interference channels and X channels, the concept of {\em Interference Alignment} is the key idea to achieve the optimal DoF. This technique, which was originally introduced in~\cite{maddah_x_tech1,maddahali2008com} and followed by~\cite{jafar2008dfr}, is based on  managing the interference to be less severe. at the receivers. In fact, the interference arriving from different transmitters are aligned at the receiver, such that it occupies the minimum number of signaling dimensions. In~\cite{cadambe2008iaa}, this idea is used to show that in the time-varying interference channels with $K$ users, the DoF of the system is $\frac{K}{2}$. In addition, in~\cite{cadambe2008dfw}, it is proven that the DoF of the time-varying $X$ network with $M$ transmitters and $K$ receivers is $\frac{MK}{M+K-1}$.  The idea of~\cite{cadambe2008iaa,cadambe2008dfw}, is to extend the precoder across time and exploit the time-variation of the channel to satisfy the conditions required for interference alignment. However, if the channel is single-antenna and not time-varying or frequency-selective, then the vector-space approach used in~\cite{cadambe2008iaa,cadambe2008dfw} falls short and does not achieve the optimum DoF. In fact, since the channel is fixed across the time, then the channel parameters do not provide enough freedom to simultaneously satisfy all the conditions required for the vector-space interference alignment.

In~\cite{Bresler-Parekh-tse}, followed by \cite{ Sridharan,sridharan3llc}, the idea of interference alignment is extended from space/time/frequency dimensions to signal level dimensions. In \cite{abolfazl-shahab-amir,Etkin-Ordentlich}, it is shown the theory of Diophantine approximation can be used to align the interference based on the properties of rational and irrational numbers. It is proven that this approach yields the optimal DoF of a certain class of the time-invariant interference channels.

Finally, in \cite{Abolfazl_X,abolfazl_k_2}, it is shown that the optimal DoF of $K$-user time-invariant interference channels is greater that one, almost surely. In~\cite{abolfazl_k_2}, it is proven that the optimal DoF is indeed $\frac{K}{2}$ for almost all $K$-user interference channels.
The achievable scheme is established based on a recent version of Khintchine-Groshev theorem. The result of~\cite{abolfazl_k_2} reveals that the field of real numbers is rich enough to transform a static interference channel to a pseudo multiple-antenna system and mimic the vector-space approaches of interference alignment. This result breaks the barrier of achieving the optimal DoF of the static channels, in which the vector-space approaches fail.

The signaling scheme presented in this work, is based on the results of~\cite{abolfazl_k_2}. We note that the machinery developed in~\cite{abolfazl_k_2} is derived for the real channels. A conventional approach to deal with a complex channel is to transform it to a real channel by decomposing the real and imaginary components. In the transformed channel, the real and imaginary parts of the channel coefficients appear twice in the channel matrix. However, the approach of~\cite{abolfazl_k_2} relies on the independency of the channel coefficients. Therefore,  applying this approach for the transformed channel is not straight-forward.  In this paper, we borrow some results from Number Theory to extend this approach to the complex channels.

\section{Main Contributions}

\begin{thm}\label{thm:ach1}
In the compound broadcast channel, modeled by~\eqref{eq:channel_model}, both DoF of $\frac{MK}{M+K-1}$ is achievable, almost surely.
\end{thm}
To derive this result, we simply divide the message of each receiver into $M$ sub-messages and then each transmitter just sends one of the sub-messages. The transmission is such that the corresponding receiver can decode all of these sub-messages. However, the contribution of these sub-messages  at other receivers are aligned. The important point is that the interference alignment is guaranteed for any channel realization.
In the proposed method, there is no cooperation among the transmitters. In other words, the output of the transmitters are independent.

We note that this result is in contrast with conjecture of~\cite{Hannan_BCCM} that the DoF of the channel converges to one if the number of states becomes large. Indeed, we show that the gap between the inner-bound and outer-bound in~\cite{Hannan_BCCM} is due to the inefficiency of the achievable scheme, and not because looseness of the outer-bound. However, in this problem, similar to time-invariant $K$-user interference channel, the vector-space approaches like zero-forcing fail to achieve the optimal DoF. The scheme that we use here is based on the machinery developed in~\cite{abolfazl_k_2} to achieve the DoF of $K$-user time-invariant interference channels.

\begin{thm}\label{thm:ach2}
In the compound broadcast channel~\eqref{eq:channel_model}, if $J_r\geq M$, for $r=1,\ldots, M$, then the optimal DoF is $\frac{MK}{M+K-1}$, almost surely.
\end{thm}

This theorem states that, the achievable scheme of Theorem~\ref{thm:ach1} is optimal when the number of states for each receiver exceeds the number of transmit antennas. Note, in this case, the uncertainty of the base station about the channel gains is high.  Remember that in MIMO Broadcast channel when the channel realization is unique, we achieve the optimal DoF through cooperation among transmitters. However, in the achievable scheme of Theorem~\ref{thm:ach1}, there is no cooperation among transmitters. The important message of this theorem is that when the uncertainty of the base station about the channel is considerable, we lose the gain of cooperation at the transmitter, but still we gain the possibility of interference alignment.

\begin{thm}\label{thm:ach3}
In the compound broadcast channel~\eqref{eq:channel_model}, if $K=M$  and $J_r= 1$, for $r=1,\ldots, K-1$ and $J_M \geq M$, then the optimal DoF is $M-1+\frac{1}{M}$,  almost surely.
\end{thm}
We note that in this case, the base station has no uncertainty about the channels of receivers 1 to $M-1$. Here, we will show that the system benefits from both cooperation among transmitters and interference alignment. In the achievable scheme, we use zero-forcing precoder such that receivers 1 to $M-1$ observe no interference. However, receiver $M$ experiences interference from all the messages, sent to the other receivers. At receiver $M$, we use interface alignment to open up space for the message of this receiver.

\begin{thm}\label{thm:x}
In the compound X channel with $M$ transmitters and $K$ receivers and finite channel realizations, the optimal DoF is  $\frac{MK}{M+K-1}$, almost surely.
\end{thm}
Note that in the achievable scheme of Theorem~\ref{thm:ach1}, we treat the broadcast channel as an X channel,  and achieve the optimal DoF of the corresponding $X$ channel.  Therefore, the achievable scheme for Theorem~\ref{thm:ach1} is indeed the achievable scheme for the corresponding compound X channel. For the converse, we use the result of~\cite{cadambe2008dfw}. It is interesting to note that the uncertainty of the transmitters about the channel realization in the compound X channel does not affect the DoF of the channel. Similar result can be proven for the compound interference channels. The reason is that the uncertainty about the channel realizations sacrifices the gain of cooperation, while in both X channels and interference channels, there is no possibility of cooperation from the beginning.

\textbf{Remark:} In the next sections, we first present the proof for the real channels. However, in Section~\ref{sec:complex}, we show
 that the machinery that we use for real channels can be extended to complex channels as well. Therefore, for example the complex DoF of a complex compound X channel with $M$ transmitters and $K$ receivers is $\frac{MK}{M+K-1}$.

\section{ Interference Alignment Approach: Achievable Scheme for Theorem~\ref{thm:ach1}}\label{sec:scheme1}
In this section, we prove Theorem~\ref{thm:ach1} and  explain the scheme which achieves the DoF of $\frac{MK}{M+K-1}$, almost surely.
For the sake of simplicity, we first elaborate the scheme for the case where the base station has two transmit antennas (M=2) and there are two users $K=2$.

Assume that the base station has message $W_1$ for receiver $1$ and $W_2$ for receiver two. In this approach, $W_1$ is divided  into two independent parts $W_{11}$ and $W_{12}$, i.e. $W_1=(W_{11}, W_{12})$. $W_{11}$ will be sent through transmitter one and $W_{12}$ will be sent through transmitter two. Similarly, $W_2$ is divided into two parts $W_2=(W_{21}, W_{22})$, where $W_{21}$ will be sent through transmitter one and $W_{22}$ will be sent through transmitter two. The transmission scheme is such that the contributions of $W_{11}$ and  $W_{12}$ are almost aligned at receiver two. Note that $W_{11}$ and  $W_{12}$ are not required at the second receiver.  Similarly, the contributions of $W_{21}$ and  $W_{22}$ are aligned at the first receiver.
Therefore, receiver one observers the contributions of  $W_{11}$,  $W_{12}$, and aligned contribution of  $W_{21}$ and  $W_{22}$. Each part takes $\frac{1}{3}$ of the DoF at receiver one. Therefore, the favorite messages $W_{11}$ and  $W_{12}$ take $\frac{2}{3}$ of the DoF, while $\frac{1}{3}$ of the resource is occupied and wasted by the aligned contribution of $W_{21}$ and  $W_{22}$. Similarly, at receiver two, the messages $W_{21}$ and  $W_{22}$ take $\frac{2}{3}$ of the DoF, while $\frac{1}{3}$ of the space is occupied by the aligned contribution of $W_{11}$ and  $W_{12}$. Therefore, we can achieve the total DoF of $\frac{4}{3}$. It is important to note that each receiver has different realizations, and the alignment must hold for all realizations. To this end, $W_{rt}$, $r=1,2$, $t=1,2$, itself is divided into $L_{rt}$ independent  parts as $W_{rt}=(W^{(1)}_{rt}, W^{(2)}_{rt}, \ldots, W^{(L_{rt})}_{rt})$, where $L_{rt}$ is almost the same for all $r=1,2$, $t=1,2$ and equal to $L$. Therefore, transmitter $t$ sends $L$ data sub-streams to receiver $r$.
At receiver one, the contributions of $(W^{(1)}_{21}, W^{(2)}_{21}, \ldots, W^{(L)}_{21})$ and $(W^{(1)}_{22}, W^{(2)}_{22}, \ldots, W^{(L)}_{22})$ have to be aligned, but it does not matter, which sub-streams of the first sets and the second sets are aligned.
This property gives us the flexibility to align the interference terms for all channel realizations. In the proposed scheme, for different channel realizations of receiver one, the contributions of $(W^{(1)}_{21}, W^{(2)}_{21}, \ldots, W^{(L)}_{21})$ and $(W^{(1)}_{22}, W^{(2)}_{22}, \ldots, W^{(L)}_{22})$ are always aligned, in a way that almost each sub-stream of the first set is aligned with a unique sub-stream from the other set. However, the mapping will change from one channel realization to the other. Similar statement is valid for receiver two.

Now the question is that how to keep the favorite sub-streams separable from each other and from the interference ones. Note that in the multiple-antenna systems, we send each data stream in a direction such that at the multiple-antenna receiver, we can separate each data stream from the others. However, here we only have single-antenna receiver and many data sub-streams. The technique is as follows:  (i) data sub-streams are modulated over integer constellations, (ii) each data sub-stream is multiplied to a particular constant which is determined by the channel coefficients. We call these coefficients as {\em modulation pseudo-vectors}. It has been shown with this approach we can separate each data sub-stream from the others under some conditions~\cite{abolfazl_k_2}.

In what follows, we step-by-step explain the proposed signaling scheme.

\subsection{Encoding}
As mentioned, in this scheme $W_r$ is divided into two independent parts $(W_{r1}, W_{r2})$ and then $W_{rt}$ is divided into $L_{rt}$ parts
$W_{rt}=(W^{(1)}_{rt}, W^{(2)}_{rt}, \ldots, W^{(L_{rt})}_{rt})$. The sub-message $W_{rt}^{(l)}$, sent by transmitter $t$, intended for receiver $r$, is encoded into the sequence $(u_{rt}^{(l)}[1], u_{rt}^{(l)}[2], \ldots, u_{rt}^{(l)}[T])$, where $T$ is the length of the codeword, and $u_{rt}^{(l)}[m]$, $m=1, \ldots, T$,  belongs to the integer constellation $(-Q, Q)_{\mathds{Z}}$. The parameter $Q$ will be given later. The sequence $(u_{rt}^{(l)}[1], u_{rt}^{(l)}[2], \ldots, u_{rt}^{(l)}[T])$ is weighted by (multiplied to) a real number $\nu^{(l)}_{rt}$, which is called modulation pseudo-vector. Each transmitter sends a weighted linear combination of the corresponding codewords. More precisely, the transmit signal by transmitters one and two are given by,
\begin{eqnarray}
x_1[m]= \lambda  \sum_{l=1}^{L_{11}} \nu^{(l)}_{11} u^{(l)}_{11}[m] + \lambda \sum_{l=1}^{L_{21}} \nu^{(l)}_{21} u^{(l)}_{21}[m]  ,\\
x_2[m]= \lambda  \sum_{l=1}^{L_{12}} \nu^{(l)}_{21} u^{(l)}_{21}[m]+ \lambda \sum_{l=1}^{L_{22}} \nu^{(l)}_{22} u^{(l)}_{22}[m].
\end{eqnarray}
 The normalizing constant $\lambda$ is chosen such that the power constraint is satisfied.

As we will see later, the modulation pseudo-vectors have two roles in the signaling: (i) it allows each receiver to separate favorite data sub-streams from the interference ones. (ii) it enables us to align the interference sub-streams at each receiver and improve the achieved DoF.

\subsection{ Choosing the  Modulation Pseudo-Vectors }

Let us define the set $\mathcal{B}_r$, for $r=1,2$, as follows:
\begin{eqnarray}
\mathcal{B}_{r}=\left \{ \prod_{r'=1, r' \neq r}^{K} \prod_{s=1}^{J_{r'}} \prod_{t=1}^{M} (h_{r't}^{\{s\}} )^{ \alpha_{r't}^{\{s\}}  }, \quad   1 \leq \alpha_{r't}^{\{s\}} \leq n_r, \quad r'\neq r   \right \},
\end{eqnarray}
where $n_r$ is a constant number. We define $L_r$ as the $|\mathcal{B}_r|$. It is easy to see that $L_1=n_1^{2J_2}$ and $L_2=n_2^{2J_1}$.
We use $\mathcal{B}_{1}$ and $\mathcal{B}_{2}$ as the set of the modulation pseudo-vectors, such that
\begin{eqnarray}
\left \{ \nu^{(l)}_{11}, \quad l=1,\ldots, L_{11} \right \}  = \left \{ \nu^{(l)}_{12},  \quad l=1,\ldots, L_{12} \right \} = \mathcal{B}_{1},\\
\left \{ \nu^{(l)}_{21}, \quad l=1,\ldots, L_{21} \right \}  = \left \{ \nu^{(l)}_{22},  \quad l=1,\ldots, L_{22} \right \} = \mathcal{B}_{2}.
\end{eqnarray}
 Since$\nu^{(l)}_{rt} =\nu^{(l')}_{rt}$ for $l \neq l'$, we have $L_{11}=L_{12}=L_1=n_1^{2J_2}$ and $L_{21}=L_{22}=L_2=n_2^{2J_1}$.  Consider $L$ as a large enough integer. We  choose $n_1$ and $n_2$ as
\begin{eqnarray}
n_1= \lfloor  L^{\frac{1}{2J_2}} \rfloor, \\
n_2= \lfloor  L^{\frac{1}{2J_1}}  \rfloor.
\end{eqnarray}
With these choices of $n_1$ and $n_2$, $L_1$ and $L_2$ are relatively close to $L$. In fact, $L_r=L+o(L)$, for $r=1,2$.

\subsection{Received Signals and Interference Alignment}
Let us focus on receiver one, when the channel state is $\mathbf{h}_1^{\{\hat{s}\}}$,  $\hat{s} \in \{ 1,\ldots, J_1\}$. The received signal is given by,
\begin{align}
y_1^{\{ \hat{s} \}}[m]& = \mathbf{h}_1^{{\{ \hat{s} \}}^{\dagger}} \mathbf{x}[m]+z_1^{\{ \hat{s}\}}[m] =\\
& =  h_{11}^{\{ \hat{s} \}} x_1[m]+ h_{12}^{\{\hat{s}\}} x_2[m]+z_1^{\{ \hat{s} \}}[m] = \\
& =  \lambda h_{11}^{\{ \hat{s} \}} \left( \sum_{l=1}^{L_{1}} \nu^{(l)}_{11} u^{(l)}_{11}[m] + \sum_{l=1}^{L_{2}} \nu^{(l)}_{21} u^{(l)}_{21}[m] \right)
\\& + \lambda h_{12}^{\{ \hat{s} \}} \left( \sum_{l=1}^{L_{1}} \nu^{(l)}_{12} u^{(l)}_{12}[m] + \sum_{l=1}^{L_{2}} \nu^{(l)}_{22} u^{(l)}_{22}[m] \right) +z_1^{\{ \hat{s}\}}[m] \\
&
=  \lambda  \left( \sum_{l=1}^{L_{1}} h_{11}^{\{ \hat{s} \}}\nu^{(l)}_{11} u^{(l)}_{11}[m] +  \sum_{l=1}^{L_{1}} h_{12}^{\{ \hat{s} \}}\nu^{(l)}_{12} u^{(l)}_{12}[m]  \right)
\\& + \lambda  \left(
  \sum_{l=1}^{L_{2}} h_{11}^{\{ \hat{s} \}}\nu^{(l)}_{21} u^{(l)}_{21}[m]
 + \sum_{l=1}^{L_{2}}h_{12}^{\{ \hat{s} \}} \nu^{(l)}_{22} u^{(l)}_{22}[m] \right) +z_1^{\{ \hat{s}\}}[m].
\end{align}
Note that the first two summations in the RHS of the above equations convey information for receiver one, while the last two summations are just interference.

It is easy to see that,
\begin{eqnarray}
h_{11}^{\{ \hat{s} \}} \nu^{(l)}_{11} \in h_{11}^{\{ \hat{s} \}}. \mathcal{B}_1 = \left \{ h_{11}^{\{ \hat{s} \}} \prod_{s=1}^{J_{2}} \prod_{t=1}^{M} (h_{2t}^{\{s\}} )^{ \alpha_{2t}^{\{s\}}  }, \quad   1 \leq \alpha_{2t}^{\{s\}} \leq n_1  \right \},\\
h_{12}^{\{ \hat{s} \}} \nu^{(l)}_{12} \in h_{12}^{\{ \hat{s} \}}. \mathcal{B}_1 = \left \{ h_{12}^{\{ \hat{s} \}} \prod_{s=1}^{J_{2}} \prod_{t=1}^{M} (h_{2t}^{\{s\}} )^{ \alpha_{2t}^{\{s\}}  }, \quad   1 \leq \alpha_{2t}^{\{s\}} \leq n_1   \right \},\\
h_{11}^{\{ \hat{s} \}} \nu^{(l)}_{21} \in h_{11}^{\{ \hat{s} \}}. \mathcal{B}_2 = \left \{ h_{11}^{\{ \hat{s} \}} \prod_{s=1}^{J_{1}} \prod_{t=1}^{M} (h_{1t}^{\{s\}} )^{ \alpha_{1t}^{\{s\}}  }, \quad   1 \leq \alpha_{1t}^{\{s\}} \leq n_2   \right \},\\
h_{12}^{\{ \hat{s} \}} \nu^{(l)}_{22} \in h_{12}^{\{ \hat{s} \}}. \mathcal{B}_2 = \left \{ h_{12}^{\{ \hat{s} \}} \prod_{s=1}^{J_{1}} \prod_{t=1}^{M} (h_{1t}^{\{s\}} )^{ \alpha_{1t}^{\{s\}}  }, \quad   1 \leq \alpha_{1t}^{\{s\}} \leq n_2   \right \}.
\end{eqnarray}

Regarding the coefficients of the received signal, we observe three important properties:
\begin{itemize}
\item[(i)] Since $h_{11}^{\{ \hat{s} \}} \neq h_{21}^{\{ \hat{s} \}}$, almost surely, then
\begin{eqnarray}
h_{11}^{\{ \hat{s} \}}. \mathcal{B}_1 \cap h_{12}^{\{ \hat{s} \}}. \mathcal{B}_1 =\emptyset.
\end{eqnarray}
Therefore, $|h_{11}^{\{ \hat{s} \}}. \mathcal{B}_1 \cup h_{12}^{\{ \hat{s} \}}. \mathcal{B}_1| =2 L_1$.
This means that at receiver one, $2L_1$ favorite data sub-streams are received with distinct coefficients.

\item[(ii)] It is easy to see that
\begin{eqnarray}
(h_{11}^{\{ \hat{s} \}}. \mathcal{B}_1 \cup h_{12}^{\{ \hat{s} \}}. \mathcal{B}_1 ) \bigcap ( h_{11}^{\{ \hat{s} \}}. \mathcal{B}_2 \cup h_{12}^{\{ \hat{s} \}}. \mathcal{B}_2) = \emptyset.
\end{eqnarray}
This means that interference sub-streams are received at receiver one with coefficients which are different from the coefficients of the favorite sub-streams.

\item[(iii)] Now let us focus on the coefficients of interference sub-streams at receiver one. It is easy to see that
 \begin{eqnarray}
 | h_{11}^{\{ \hat{s} \}}. \mathcal{B}_2 \cup h_{12}^{\{ \hat{s} \}}. \mathcal{B}_2 |=(n_2)^{2J_1-2}(n_2+1)^2.
 \end{eqnarray}
 Remember that $| h_{11}^{\{ \hat{s} \}}. \mathcal{B}_2| = |h_{12}^{\{ \hat{s} \}}. \mathcal{B}_2| =(n_2)^{2J_1}$. This means that $|h_{11}^{\{ \hat{s} \}}. \mathcal{B}_2 \cup h_{12}^{\{ \hat{s} \}}. \mathcal{B}_2 |$ has almost the same cardinality as $| h_{11}^{\{ \hat{s} \}}. \mathcal{B}_2|$ and $|h_{12}^{\{ \hat{s} \}}. \mathcal{B}_2|$. It
 implies that the set $ h_{11}^{\{ \hat{s} \}}. \mathcal{B}_2$ and the set $ h_{12}^{\{ \hat{s} \}}. \mathcal{B}_2$ are almost the same  with just few different elements (compared to the size of each set). In other words, the sub-streams, sent by transmitters one and two, intended for receiver two, are received at receiver one with the same coefficients. In fact, this property results in the alignment of  interference. Note that this property holds for all channel realizations.
\end{itemize}

At receiver one, we merge the interference sub-streams with the similar coefficients, so we have
\begin{eqnarray}
y_1^{\{ \hat{s} \}}[m]&& =  \lambda  \left( \sum_{l=1}^{L_{1}} h_{11}^{\{ \hat{s} \}}\nu^{(l)}_{11} u^{(l)}_{11}[m] +  \sum_{l=1}^{L_{1}} h_{12}^{\{ \hat{s} \}}\nu^{(l)}_{12} u^{(l)}_{12}[m] +
 \sum_{l=1}^{\kappa_2} \hat{\nu}^{(l)}_{1,\hat{s}} \bar{u}^{(l)}_{1,\hat{s}}[m]
 \right) +z_1^{\{ \hat{s}\}}[m],
\end{eqnarray}
where
$\kappa_2 = n_2^{2J_1-2}(n_2+1)^2$, and $\hat{\nu}^{(l)}_{1,\hat{s}} \in h_{11}^{\{ \hat{s} \}}. \mathcal{B}_2 \cup h_{12}^{\{ \hat{s} \}}. \mathcal{B}_2$. In addition, $\bar{u}^{(l)}_{1,\hat{s}}[m] \in (-2Q,2Q)_{\mathds{Z}}$. Therefore, we have a noisy version of the integer combination of $\kappa_2+2L_1$ real numbers. It is important to note that  these numbers are monomial functions of the channel coefficients, where these functions are linearly independent.
Note that the fraction of $\frac{2L_1}{\kappa_2+2L_1}$ of the arrived data sub-streams with different coefficients are favorite sub-streams. Since $\kappa_2= L+o(L)$ and $L_1=L+o(L)$, then $\frac{2L_1}{\kappa_2+2L_1} \simeq \frac{2}{3}$.

Similarly, at receiver two, where the channel state is $\mathbf{h}_2^{\{\hat{s}\}}$, where $\hat{s} \in \{ 1,\ldots, J_2\}$, we have
\begin{eqnarray}
y_2^{\{ \hat{s} \}}[m]=  \lambda  \left( \sum_{l=1}^{L_{2}} h_{21}^{\{ \hat{s} \}}\nu^{(l)}_{21} u^{(2l)}_{l}[m] +  \sum_{l=1}^{L_{2}} h_{22}^{\{ \hat{s} \}}\nu^{(l)}_{22} u^{(l)}_{22}[m] + \sum_{l=1}^{\kappa_1} \hat{\nu}^{(l)}_{2,\hat{s}} \bar{u}^{(l)}_{2,\hat{s}}[m]
 \right) +z_2^{\{ \hat{s}\}}[m],
\end{eqnarray}
where $\kappa_1 = n_1^{2J_2-2}(n_1+1)^2$ and $\hat{\nu}^{(l)}_{2,\hat{s}} \in h_{21}^{\{ \hat{s} \}}. \mathcal{B}_1 \cup h_{22}^{\{ \hat{s} \}}. \mathcal{B}_1$. In addition, $\bar{u}^{(l)}_{2,\hat{s}}[m] \in (-2Q,2Q)_{\mathds{Z}}$. Therefore, we have a noisy version of the integer combination of $\kappa_1+2L_2$ real numbers. Again, it is important to note that these numbers are monomial functions of the channel coefficients where these functions are linearly independent.  Again $\kappa_1= L+o(L)$ and $L_2=L+o(L)$, then $\frac{2L_2}{\kappa_1+2L_2} \simeq \frac{2}{3}$ of the separable data sub-streams convey favorite messages.

Note that at each receiver, the total available DoF is just one. Here we try to develop the signaling scheme such that each data sub-stream has DoF of $\frac{1}{\xi}$ DoF, where
\begin{eqnarray}
\xi= \max\{\kappa_2+2L_1 , \kappa_1+2L_2\}.
\end{eqnarray}
Therefore, at receiver $1$, we have $\frac{2L_1}{\xi}$ portion of the available DoF is used for receiving the favorite data sub-streams, while   $\frac{\kappa_2}{\xi}$ is wasted for interference. It is easy to see that $\frac{2L_1}{\xi}$ is almost $\frac{2}{3}$, while $\frac{\kappa_2}{\xi}$ is almost $\frac{1}{3}$. Similar statement is valid for the second receiver.

In the next subsection, we design the parameters of the signaling scheme.

\subsection{Choosing $Q$ and $\lambda$}
Now we choose $Q$ as follows:
\begin{eqnarray}
Q = (\frac{P}{2})^{\frac{1-\epsilon}{2(\xi + \epsilon) }},
\end{eqnarray}
where $\epsilon$ is an arbitrary small constant. Note that $u^{(l)}_{rt}[m]$ is from the integer constellation in  $(-Q,Q)$, where the rate of this constellation is $\log_2 (2Q)= \frac{1-\epsilon}{2(\xi + \epsilon)} \log_2(\frac{P}{2})+ 1$.

In addition, we have,
\begin{align}
\mathbb{E} [x^2_t[m]]&=\lambda^2 \left(\sum_{l=1}^{L_{1}} \left(\nu^{(l)}_{1t}\right)^2 \mathbb{E} [ (u^{(l)}_{1t}[m])^2 ]+ \sum_{l=1}^{L_{2}} \left(\nu^{(l)}_{2t} \right)^2 \mathbb{E} [ (u^{(l)}_{2t}[m])^2] \right) \\
&\leq \lambda^2 Q^2 \left(\sum_{l=1}^{L_{1}} \left(\nu^{(l)}_{1t}\right)^2 +  \sum_{l=1}^{L_{2}} \left(\nu^{(l)}_{2t} \right)^2  \right),
\end{align}
where we use the independency of the data sub-streams.
Wee choose $\lambda$ such that
\begin{eqnarray}
\mathbb{E} [x^2_t[m]] \leq 0.5P.
\end{eqnarray}
One choice for $\lambda$ is
\begin{eqnarray}
\lambda =\frac{(.5P)^{\frac{1}{2}}}{\Gamma Q},
\end{eqnarray}
where
\begin{eqnarray}
\Gamma^2= \sum_{l=1}^{L_{1}} \left(\nu^{(l)}_{11}\right)^2 +  \sum_{l=1}^{L_{2}} \left(\nu^{(l)}_{21} \right)^2 =
\sum_{\nu \in \mathcal{B}_1} \nu^2 + \sum_{\nu \in \mathcal{B}_2} \nu^2  .
\end{eqnarray}
Then
\begin{eqnarray}
\lambda= \frac{1}{\Gamma} (.5P)^{ \frac{\xi-1+2\epsilon}{2(\xi+\epsilon)}}.
\end{eqnarray}

\subsection{Constellation Formed At Each Receiver}
At receiver one, when the channel state is $\mathbf{h}_1^{\{\hat{s}\}}$,  $\hat{s} \in \{ 1,\ldots, J_1\}$,  the received signal at each time $m$ is a noisy version of a point from the constellation $\mathcal{C}_1$, where
\begin{eqnarray*}
\mathcal{C}_1 = \left\{  \lambda  \left( \sum_{l=1}^{L_{1}} h_{11}^{\{ \hat{s} \}}\nu^{(l)}_{11} u^{(l)}_{11} +  \sum_{l=1}^{L_{1}} h_{12}^{\{ \hat{s} \}}\nu^{(l)}_{12} u^{(l)}_{12} +
 \sum_{l=1}^{\kappa_1} \hat{\nu}^{(l)}_{1,\hat{s}} \bar{u}^{(l)}_{1,\hat{s}}
 \right), \
u^{(l)}_{11}, u^{(l)}_{21} \in (-Q,Q)_{\mathds{Z}}, \ \bar{u}^{(l)}_{1,\hat{s}} \in (-2Q,2Q)_{\mathds{Z}}
\right\}.
\end{eqnarray*}
Using the Theorem 4 of~\cite{abolfazl_k_2}, we can show that the minimum distance of this constellation is $\eta \left( \frac{P}{2} \right)^{\epsilon}$, almost surely, where $\eta$ is a constant, independent of $P$.

This means that
\begin{itemize}
\item[(i)] There is a one to one mapping between $\left(\left\{ u^{(l)}_{1t} \right\}_{\begin{smallmatrix} t=1,\ldots,2 \\ l=1,\ldots,L_1 \end{smallmatrix}},  \left\{\bar{u}^{(l)}_{1,\hat{s}}\right\}_{\begin{smallmatrix}  l=1,\ldots,\kappa_{2} \end{smallmatrix} } \right)$ and the points of the constellation $\mathcal{C}_1$.
\item[(ii)] In high power, we can de-noise the received signal, the detect the point of the constellation $\mathcal{C}_r$ with vanishing probability of error, find the unique corresponding  $\left(\left\{ u^{(l)}_{1t} \right\}_{\begin{smallmatrix} t=1,\ldots,2 \\ l=1,\ldots,L_1 \end{smallmatrix}},  \left\{\bar{u}^{(l)}_{1,\hat{s}}\right\}_{\begin{smallmatrix}  l=1,\ldots,\kappa_{2} \end{smallmatrix} } \right)$.
\end{itemize}

The same statement is true for receiver two. Therefore, at receiver $r$ and at each time $m$,  we apply hard detection to find $\left(\left\{ \hat{u}^{(l)}_{rt}[m] \right\}_{\begin{smallmatrix} t=1,\ldots,2 \\ l=1,\ldots,L_r \end{smallmatrix}},  \left\{\hat{\bar{u}}^{(l)}_{r,\hat{s}}[m]\right\}_{\begin{smallmatrix}  l=1,\ldots,\kappa_{\hat{r}}, \ \hat{r}\neq r\end{smallmatrix} } \right)$. Then, we pass the sequence $\left(\left\{ \hat{u}^{(l)}_{rt}[m] \right\}_{\begin{smallmatrix} t=1,\ldots,2 \\ l=1,\ldots,L_r \end{smallmatrix}},  \left\{\hat{\bar{u}}^{(l)}_{r,\hat{s}}[m]\right\}_{\begin{smallmatrix}  l=1,\ldots,\kappa_{\hat{r}}, \ \hat{r}\neq r\end{smallmatrix} } \right)$ to the decoder to decode $\hat{W}_{rt}^{l}$, for $l=1,\ldots, L_r$ and $t=1,2$.

\subsection{Performance Analysis}
As mentioned, at receiver $r$ and at each time, we use hard detection to detect \\
$\left(\left\{ \hat{u}^{(l)}_{rt}[m] \right\}_{\begin{smallmatrix} t=1,\ldots,2 \\ l=1,\ldots,L_r \end{smallmatrix}},  \left\{\hat{\bar{u}}^{(l)}_{r,\hat{s}}[m]\right\}_{\begin{smallmatrix}  l=1,\ldots,\kappa_{\hat{r}}, \ \hat{r}\neq r\end{smallmatrix} } \right)$. Probability of error, $P_e$ of this detection is upper-bounded by
$P_e \leq Q\left( \frac{d_{min}}{2} \right) =  Q\left( \eta \left(\frac{P}{2} \right)^{\epsilon} \right) $, and therefore, $P_e \rightarrow 0$ as $P \rightarrow \infty$. Then, using the fact that $u^{(l)}_{rt}[m]$ is from the integer constellation in  $(-Q,Q)$ with rate $\log_2 (2Q)= \frac{1-\epsilon}{2(\xi + \epsilon)} \log_2(\frac{P}{2})+ 1$, we can show that each of the data sub-streams $u^{(l)}_{r1}[m]$ and $u^{(l)}_{r2}[m]$ achieves the DoF of $\frac{1-\epsilon}{\xi + \epsilon}$. Therefore, we achieve the DoF of $2L_1 \frac{1-\epsilon}{\xi + \epsilon}$ at receiver one and $2L_2 \frac{1-\epsilon}{\xi + \epsilon}$ at receiver two. Therefore, this scheme achieves the total DoF of
\begin{eqnarray}
2(L_1+L_2)\frac{1-\epsilon}{\xi + \epsilon}.
\end{eqnarray}
Since $L_1$ and $L_2$ are in the order of $L+o(L)$, and $\xi =3L+o(L)$,  by choosing large enough $L$ and small enough $\epsilon$, this scheme archives the DoF arbitrary close to $\frac{4}{3}$.

We can apply the same approach for the general compound broadcast channel~\eqref{eq:channel_model}, and achieve the DoF of $\frac{MK}{M+K-1}$. Refer to Appendix~\ref{app:1} for details.

\textbf{Remark:} Note that in the achievable scheme, we assume no cooperation among transmitters. This means that we treat the channel as a compound X channel. Therefore, the achievable scheme of Theorem~\ref{thm:ach1} is an achievable scheme for the corresponding compound X channel. Therefore, for the compound X channel, we achieve the DoF of $\frac{MK}{M+K-1}$.  The converse of Theorem~\ref{thm:x} is proven by using the upper-bound of the DoF of the single-state X network given in~\cite{cadambe2008dfw}. We can use similar argument to prove that the DoF of the finite-state compound K-user interference channels is $\frac{K}{2}$, almost surely.

\section{Combination of Transmit Cooperation and Interference Alignment: Achievable Scheme for Theorem~\ref{thm:ach3} }~\label{sec:scheme2}
In this section, we focus on the case, where $K=M$, $J_r=1$, for $r=1,\ldots, M-1$, and $J_M \geq M$. In this case, the channel model is simplified to,
\begin{align}\label{eq:channel_mode2}
y_r[m]=&\mathbf{h}_r^{\dagger} \mathbf{x}[m]+z_r[m], \quad r=1,\ldots,M-1,\\
y_M^{\{s\}}[m]=&\mathbf{h}_M^{{\{s\}}^{\dagger}} \mathbf{x}[m]+z_M^{\{s\}}[m],  \quad s=1,\ldots,J_M.
\end{align}
Apparently, by using the scheme presented in the previous section, the DoF of $\frac{M^2}{2M-1}$ is achievable. Remember that in the scheme of Section \ref{sec:scheme1}, the possibility of cooperation among the transmitters is simply ignored. Indeed, we will show that for some cases where the uncertainty of the transmitter about the channel states is considerable, the approach of Section \ref{sec:scheme1} is optimal and therefore ignoring the possibility of cooperation does not affect the achievable DoF. However, in the channel \eqref{eq:channel_mode2}, the knowledge of the base station about the channel states is considerable. In fact, the base station knows the perfect channel state for receivers $1$ to $M-1$. This knowledge allows us to improve the DoF by exploiting the possibility of the cooperation among transmitters.

 In this section, we propose a signaling scheme which is based on a combination of both zero-forcing and interference alignment. The proposed  scheme achieves the DoF of $M-1+\frac{1}{M}$. More precisely,  receiver $r$, $1 \leq r \leq M-1$, achieves the  DoF of one, while the last receiver achieves the DoF of $\frac{1}{M}$. Indeed, we develop a special form of zero-forcing precoder,  such that receivers $1$ to $M-1$ do not experience any interference. In contrary, receiver $M$ observes interference from the data sent to all other receivers. However, we use interference alignment to reduce the constructive effect of the interference and guarantee the DoF of $\frac{1}{M}$ for receiver $M$.

The outline of the alignment scheme is as follows. Let $W_r$ be the message for receiver $r$. Here, message $W_r$,  for $r=1,\ldots,M-1$, is divided into $M$ independent messages, as $W_r=(W_{r1},\ldots,W_{rM})$. In the proposed scheme, the contributions of  $W_{r1},\ldots,W_{rM}$ at receiver $M$ are aligned and occupy $\frac{1}{M}$ of the available DoF. Therefore in total, $\frac{M-1}{M}$ of the available DoF at receiver $M$ is occupied by interference, and the rest is used to receive the favorite message $W_M$.

Here, we elaborate the proposed scheme step by step.

\subsection{Zero-Forcing Precoders}
For receiver $r$, the base station uses the precoding matrix $\mathbf{V}^{[r]}=[\mathbf{v}_1^{[r]},\ldots, \mathbf{v}_M^{[r]} ] \in \mathds{R}^{M\times M} $. The columns of $\mathbf{V}^{[r]}$ are selected randomly from the subspace, which is orthogonal to $\textrm{Span}\{ \mathbf{h}_{\hat{r}} , \hat{r}=1,\ldots, M-1, \hat{r} \neq r \}$. Therefore $ \mathbf{v}^{[r]}_{i} \bot \mathbf{h}_{\hat{r}}$, for $r\neq \hat{r}$.

For receiver $M$, we choose the precoding vector $\mathbf{v}^{[M]} \in \mathds{R}^{M\times 1}$, orthogonal to $\textrm{Span}[\mathbf{h}_1,\ldots,\mathbf{h}_{M-1}]$. These precoding matrices guarantee interference-free signals for receivers $1$ to $M-1$.

\subsection{Encoding}
We define $g_{ri}^{\{s\}}$ for  $s=1,\ldots,J_M$, $r=1,\ldots,M-1$, and $i=1,\ldots,M$, as
\begin{eqnarray}\label{eq:g}
g_{ri}^{\{s\}}=(\mathbf{h}_M^{\{s\}})^{\dagger} \mathbf{v}_i^{[r]}.
\end{eqnarray}
Later we use $g_{ri}^{\{s\}}$ to construct modulation pseudo-vectors.

As mentioned, the message for receiver $r$, $1 \leq r \leq M-1$, is deviled into $M$ independent parts, i.e. $W_r=(W_{r1},\ldots,W_{rM})$.
Then, $W_{ri}$ itself is divided into $L_r$ independent parts, $W_{ri}=(W_{ri}^{(1)}, W_{ri}^{(2)},\ldots, W_{ri}^{(L_r)} )$.
$W^{(l)}_{ri}$ is coded into the sequence $(u_{ri}^{(l)}[1], u_{ri}^{(l)}[2], \ldots, u_{ri}^{(l)}[T])$, where $u_{ri}^{(l)}[m]$, $m=1, \ldots, T$,  belongs to the integer constellation $(-Q, Q)_{\mathds{Z}}$.  The string $(u_{ri}^{(l)}[1], u_{ri}^{(l)}[2], \ldots, u_{ri}^{(l)}[T])$ is weighted by the modulation pseudo-vector $\nu^{(l)}_{ri}$. We form the weighted linear combination of the corresponding codewords as,
\begin{eqnarray}
\omega^{[r]}_i[m]= \lambda  \sum_{l=1}^{L_{r}} \nu^{(l)}_{ri} u^{(l)}_{ri}[m].
\end{eqnarray}
Then, we define the vector $\boldsymbol{w}^{[r]}[m]$, as
\begin{eqnarray}
 \boldsymbol{w}^{[r]}[m]=[\omega^{[r]}_1[m],\omega^{[r]}_2[m],\ldots,\omega^{[r]}_{M}[m]]^{\dagger},\ r=1,\ldots,M-1.
\end{eqnarray}

The encoding approach for message $W_M$ is different. Message  $W_M$  is divided into $L_M$ independent parts, as $W_{M}=(W_{M}^{(1)}, W_{M}^{(2)},\ldots, W_{M}^{(L_M)} )$.

Then, $W^{(l)}_{M}$ is coded into the sequence
$(u_{M}^{(l)}[1], u_{M}^{(l)}[2], \ldots, u_{M}^{(l)}[T])$, where again $u_{M}^{(l)}[m]$, $m=1, \ldots, T$,  belongs to the integer constellation $(-Q, Q)_{\mathds{Z}}$.  Then we form the following weighted linear combination using  the modulation pseudo-vectors $\nu^{(l)}_{M}$.
\begin{eqnarray}
\omega^{[M]}[m]= \lambda   \sum_{l=1}^{L_M} \nu^{(l)}_{M} u^{(l)}_{M}[m].
\end{eqnarray}
The transmit vector $\mathbf{x}[m]$ is equal to
\begin{eqnarray}
\mathbf{x}[m]=\sum_{r=1}^{M-1} \mathbf{V}^{[r]} \boldsymbol{w}^{[r]}[m]+\mathbf{v}^{[M]}\omega^{[M]}[m].
\end{eqnarray}

\subsection{Modulation Pseudo-Vectors for User $r$, $1\leq  r \leq M-1$}

Let us define the set $\mathcal{B}_r$, for $r=1,\ldots, M-1$, as follows:
\begin{eqnarray}
\mathcal{B}_{r}=\left \{  \prod_{s=1}^{J_{M}} \prod_{i=1}^{M} (g_{ri}^{\{s\}} )^{ \alpha_{ri}^{\{s\}}  }, \quad   1 \leq \alpha_{ri}^{\{s\}} \leq n_r, \right \},
\end{eqnarray}
where $g_{ri}^{\{s\}}$ is defined in \eqref{eq:g} and  $n_r$ is a constant number.

We use $\mathcal{B}_{r}$  as the set of the modulation pseudo-vectors for the data sub-streams intended for receiver $r$,  i.e.,
\begin{eqnarray}
\left \{ \nu^{(l)}_{rt}, \quad l=1,\ldots, L_{r} \right \}  = \mathcal{B}_{r}.
\end{eqnarray}
Note that $\nu^{(l)}_{rt} =\nu^{(l')}_{rt}$, if $l \neq l'$. Therefore,  $L_r$ is equal to $|\mathcal{B}_r|$. We choose $n_1=n_2=\ldots=n_{M-1}=n$, and therefore $L_r=L$, where $L=n^{MJ_M}$.

\subsection{Modulation Pseudo-Vectors for User $M$}
Let $\beta$ be a randomly-chosen real number. we choose $\nu^{(l)}_{M}$ as
\begin{eqnarray}
\nu^{(l)}_{M}=\beta^{l}, \ l=1,\ldots,L_M.
\end{eqnarray}
Set $L_M=L=n^{MJ_M}$.

\subsection{Received Signal at Receiver $r$, $1\leq r \leq M-1$}
It is easy to see that,
\begin{align}
y_r[m]& = \mathbf{h}_r^{\dagger} \mathbf{x}[m]+z_r[m] \\
&=\mathbf{h}_r^{\dagger} \left(\sum_{\hat{r}=1}^{M-1} \mathbf{V}^{[\hat{r}]} \boldsymbol{\omega}^{[\hat{r}]}[m]+\mathbf{v}^{[M]}\omega^{[M]}[m] \right)+z_r[m]\\
& = \mathbf{h}_r^{\dagger} \mathbf{V}^{[r]} \boldsymbol{w}^{[r]}[m]+z_r[m] \\
& = \sum_{i=1}^{M}\mathbf{h}_r^{\dagger} \mathbf{v}_i^{[r]} {\omega}^{[r]}[m]+z_r[m] \\
&= \lambda  \sum_{i=1}^{M}   \sum_{l=1}^{L}  \mathbf{h}_r^{\dagger} \mathbf{v}_i^{[r]}\nu^{(l)}_{ri} u^{(l)}_{ri}[m]+z_r[m],
\end{align}
where we use the fact that  $\mathbf{h}_r$ is orthogonal to the columns of $\mathbf{V}^{[\hat{r}]}$, where $r\neq \hat{r}$, and also $\mathbf{h}_r^{\dagger} \bot \mathbf{v}^{[M]}$.

Note that receiver $r$ does not observe any interference. In addition, it is easy to see that  $\mathbf{h}_r^{\dagger}\mathbf{v}_i^{[r]}\nu^{(l)}_{ri}$ are monomial functions of $\mathbf{h}_r^{\dagger}\mathbf{v}_i^{[r]}$, and $g_{ri}^{(s)}$, for different $i$'s and $s$'s, where these monomial functions are linearly independent.
Therefore, $y_r[m]$ is a noisy version of an integer combination of $ML$ real numbers.

\subsection{Received Signal at Receiver $M$ with Channel $\mathbf{h}_{M}^{\{ \hat{s}\}}$}
The received signal at receiver $M$, where the channel state is $\mathbf{h}_{M}^{\{ \hat{s}\}}$ is given by,
\begin{align}
y_M^{\{ \hat{s} \}}[m]& = \mathbf{h}_M^{{\{ \hat{s} \}}^{\dagger}} \mathbf{x}[m]+z_M^{\{ \hat{s}\}}[m] =\\
& =  \mathbf{h}_M^{{\{ \hat{s} \}}^{\dagger}} \left(\sum_{r=1}^{M-1} \mathbf{V}^{[r]} \boldsymbol{\omega}^{[r]}[m]+\mathbf{v}^{[M]}\omega^{[M]}[m] \right)+z_M^{\{ \hat{s}\}}[m]\\
& = \sum_{r=1}^{M-1}\sum_{i=1}^{M}\mathbf{h}_M^{{\{ \hat{s} \}}^{\dagger}} \mathbf{v}_i^{[r]} {\omega}^{[r]}[m] + \mathbf{h}_M^{{\{ \hat{s} \}}^{\dagger}} \mathbf{v}^{[M]}\omega^{[M]}[m] +z_M^{\{ \hat{s}\}}[m]\\
& = \sum_{r=1}^{M-1}\sum_{i=1}^{M} g_{ri}^{\{ \hat{s} \}}  {\omega}^{[r]}[m] + \mathbf{h}_M^{{\{ \hat{s} \}}^{\dagger}} \mathbf{v}^{[M]}\omega^{[M]}[m] +z_M^{\{ \hat{s}\}}[m]\\
& =
\lambda \sum_{r=1}^{M-1}\sum_{i=1}^{M} \sum_{l=1}^{L} g_{ri}^{\{ \hat{s} \}}\nu^{(l)}_{ri} u^{(l)}_{ri}[m] +  \lambda \sum_{l=1}^{L} \mathbf{h}_M^{{\{ \hat{s} \}}^{\dagger}} \mathbf{v}^{[M]} \nu^{(l)}_{M} u^{(l)}_{M}[m] +z_M^{\{ \hat{s}\}}[m],
\end{align}
where we use the definition $g_{ri}^{\{ \hat{s} \}}$ in \eqref{eq:g}.

It is easy to see that
\begin{eqnarray}
|\cup_{i=1}^{M} g_{ri}^{\{ \hat{s} \}} . \mathcal{B}_r| =\kappa,
\end{eqnarray}
where $\kappa= n^{M(J_M-1)}(n+1)^M$. This means that
 $|\cup_{i=1}^{M} g_{ri}^{\{ \hat{s} \}} . \mathcal{B}_r|$ is almost the same as $| g_{ri}^{\{ \hat{s} \}} . \mathcal{B}_r|=n^{MJ_L}$. It implies that
 the sets $ g_{ri}^{\{ \hat{s} \}} . \mathcal{B}_r$, $i=1,\ldots,M$,  are almost the same. Therefore, the data sub-streams, carrying $W_{ri}$ to $W_{rM}$,  are arrived with similar real coefficients. This property implies that the contributions of $W_{r1}$ to $W_{rM}$ are aligned at receiver $M$.  It is important to note that this property holds, irrespective to the channel state of the last receiver. Note that the favorite data-streams are received with the real coefficients which are distinct from those of the interference sub-streams.

Let us merge the the interference sub-streams with the same real coefficients. Therefore, we have
\begin{align}
y_M^{\{ \hat{s} \}}[m]=
\lambda \sum_{r=1}^{M-1} \sum_{l=1}^{\kappa} \hat{\nu}^{(l)}_{r,\hat{s}} \bar{u}^{(l)}_{r,\hat{s}}[m] +  \lambda \sum_{l=1}^{L} \mathbf{h}_M^{{\{ \hat{s} \}}^{\dagger}} \mathbf{v}^{[M]} \nu^{(l)}_{M} u^{(l)}_{M}[m] +z_M^{\{ \hat{s}\}}[m],
\end{align}
where $\hat{\nu}^{(l)}_{r,\hat{s}} \in \cup_{i=1}^{M} g_{ri}^{\{ \hat{s} \}} . \mathcal{B}_r$ and $\bar{u}^{(l)}_{r,\hat{s}}[m] \in (-MQ,MQ)$.

It is easy to see that in the above integer expansion, all the real coefficients are distinct. In addition, each coefficient is a monomial function of $g_{ri}^{\{ \hat{s} \}}$, $\mathbf{h}_M^{{\{ \hat{s} \}}^{\dagger}} \mathbf{v}^{[M]}$ and $\beta$, where these functions are linearly independent.

Let $\xi=(M-1)\kappa+L$, and then let
\begin{eqnarray}
Q = (P)^{\frac{1-\epsilon}{2(\xi + \epsilon) }},
\end{eqnarray}
where $\epsilon$ is an arbitrary small constant. Then, we have,
\begin{eqnarray}
\mathbb{E} [\| \mathbf{x}[m] \|^2]=  \sum_{r=1}^{M-1}\mathbb{E} [ (\boldsymbol{\omega}_r)^{\dagger} (\mathbf{V}^{[r]})^{\dagger} \mathbf{V}^{[r]}\boldsymbol{\omega}_r] + \mathbb{E} [ (\boldsymbol{\omega}_M)^{\dagger} (\mathbf{v}^{[M]})^{\dagger} \mathbf{v}^{[M]}\boldsymbol{\omega}_M ]
\leq Q^2\lambda^2 \Gamma^2
\end{eqnarray}
where
\begin{eqnarray}
\Gamma^2=\sum_{r=1}^{M-1}\sigma_{\max}^2( \mathbf{V}^{[r]}) \sum_{i=1}^{M}\sum_{l=1}^{L} (\nu_{ri}^{(l)})^2+  \|\mathbf{v}^{[M]}\|^2 \sum_{l=1}^{L} (\nu_{M}^{(l)})^2,
\end{eqnarray}
and $\sigma_{\max}( \mathbf{V}^{[r]})$ denotes the largest singular value of $\mathbf{V}^{[r]}$.
To satisfy the power constraint, we choose $\lambda$ as
\begin{eqnarray}
\lambda =\frac{(P)^{\frac{1}{2}}}{\Gamma Q}
\end{eqnarray}
or,
\begin{eqnarray}
\lambda= \frac{1}{\Gamma} P^{ \frac{\xi-1+2\epsilon}{2(\xi+\epsilon)}}.
\end{eqnarray}

Then, it is easy to see that all the conditions of Theorem 4 of~\cite{abolfazl_k_2} are satisfied. Therefore, receiver $M$ achieves the DoF of $L \frac{1-\epsilon}{\xi+\epsilon}$ and receiver $r$, $1 \leq r \leq M-1$, achieves the DoF of $ML \frac{1-\epsilon}{\xi+\epsilon}$.
Note that $L=n^{MJ_M}$ and $\xi=(M-1)n^{M(J_M-1)}(n+1)^M+n^{MJ_M}$, it is easy to see that $((M-1)ML+L) \frac{1-\epsilon}{\xi+\epsilon}$, can be arbitrary close to $M-1+\frac{1}{M}$, by choosing large enough $n$ and small enough $\epsilon$.

\section{Outer-Bounds: Converse for Theorems~\ref{thm:ach2} and ~\ref{thm:ach3}}
In this section, we prove the converse of Theorems~\ref{thm:ach2} and ~\ref{thm:ach3}. These results extend the outer-bounds presented in~\cite{Hannan_BCCM} following the same arguments.

First we need the following theorem.
\begin{thm}\label{thm:TD_EX}
Consider a broadcast channel $\Pr(Y_1,\ldots,Y_K|X)$ with $K$ receivers  with the degradedness property represented by $X \leftrightarrow Y_1 \leftrightarrow Y_2 \leftrightarrow ...\leftrightarrow Y_{s-1} \leftrightarrow (Y_{s},\ldots,Y_{K})$. Consider the messages
$W_1$, $W_2$,\ldots, $W_{s}$. Message $W_i$ is requires by receivers $Y_j$, $j=1,\ldots,i$, for $i=1, \ldots,s-1$, but message $W_{s}$ is required by all receivers. The rate of message $W_i$ is denoted by $R_i$. Then, the capacity region is the convex union of all the rates satisfying the following inequalities
\begin{align}
R_1 & \leq I(Y_1;X|U_1)\\
R_i & \leq I(Y_i;U_{i-1}|U_{i})\  \quad \ i=2,\ldots,s-1,\\
R_{s} &\leq I(Y_i;U_{s-1})\  \quad \ i=s,\ldots,K,
\end{align}
for some joint distributions \\
 $\Pr(u_{s-1},\ldots, u_{1},x,y_1,\ldots,y_K)=\Pr(u_{s-1}) \Pr(u_{s-2}|u_{s-1}) \ldots \Pr(u_2|u_1) \Pr(x|u_1)\Pr(y_1,\ldots,y_k|x)$.
\end{thm}

\begin{proof}
This is a direct extension of Theorem 3.1 of ~\cite{Dig_Tse}. The achievable scheme is based on the superposition coding. The outer-bound is proven using the standard arguments used in Theorem 3.1 of ~\cite{Dig_Tse}.
\end{proof}

Here, we use the above theorem to prove a key outer-bound for the compound MIMO broadcast channels.
\begin{thm}\label{thm:main}
Consider the broadcast channel \eqref{eq:channel_model}, where $J_K \geq M $, then
\begin{eqnarray}
\lim_{P \leftarrow \infty}\frac{\sum_{r=1}^{K-1}R_i+MR_K}{0.5\log_2 P} \leq M.
\end{eqnarray}
\end{thm}

\begin{proof}
It is sufficient to prove the result for the case where $J_r=1$ for $r=1,\ldots,K-1$ and $J_K=M$.  The unique realization of  receiver $r$ is denoted by $\mathbf{h}_r$, for $r=1,\ldots,K-1$. For simplicity, we denote the $M$ realizations of receiver $K$ as users $K$ to $M+K-1$, with channel $\mathbf{h}_{K}$, $\mathbf{h}_{K+1}$, $\ldots$, $\mathbf{h}_{M+K-1}$, with the corresponding noise $z_{K}$ to $z_{M+K-1}$.

We define $\mathbf{H}_r$ as $\mathbf{H}_r=[\mathbf{h}_{r}, \mathbf{h}_{r+1}, \ldots, \mathbf{h}_{M+K-1}]$, $\mathbf{z}_r$ as $\mathbf{z}_r=[z_r,\ldots,z_{M+K-1}]^{\dagger}$, and $\mathbf{y}_r$ as $\mathbf{y}_r=[y_r,\ldots,y_{M+K-1}]^{\dagger}$.

Now consider the following broadcast channel
\begin{eqnarray}
\mathbf{y}_r=\mathbf{H}_r^{\dagger}\mathbf{x}+\mathbf{z}_r, \ r=1,\ldots,K-1,\\
y_r= \mathbf{h}_r^{\dagger}\mathbf{x}+z_r, \ r=K,\ldots,M+K-1,
\end{eqnarray}
formed by giving the received signal of $y_{r+1},\ldots, y_{M+K-1}$ to receiver $r$.

The new channel has the degradedness property $\mathbf{x} \leftrightarrow \mathbf{y}_1 \leftrightarrow \ldots \leftrightarrow \mathbf{y}_{K-1}  \leftrightarrow (y_{K},\ldots,y_{M+K-1})$. Note that, if message $W_K$ can be decoded at receivers  $y_{K},y_{K+2}\ldots,y_{M+K-1}$ in the original channel, it can be decoded in all receivers of the new channel. Moreover, if message $W_r$, $1 \leq r \leq K-1$, can be decoded at receiver $r$ of the original channel, it can be decoded at receivers 1 to $r$ in the new channel. Therefore, by using Theorem~
\ref{thm:TD_EX}, there is a Markov chain of random variables $v_{K-1} \leftrightarrow v_{K-2} \ldots \leftrightarrow v_{1}   \leftrightarrow \mathbf{x}  \leftrightarrow (\mathbf{y}_1,\ldots, \mathbf{y}_K, y_{K+1},\ldots, y_{M+K-1})$ such that

\begin{align}\label{eq:1}
&R_1+\ldots+R_{K-1}+MR_K  \\ \nonumber
& \leq
I(\mathbf{x};\mathbf{y}_1|v_1) +  I(v_1;\mathbf{y}_2|v_2)+ \ldots+I(v_{K-2};\mathbf{y}_{K-1}|v_{K-1})+ \sum_{r=K}^{M+K-1} I(v_{K-1};y_r)\\ \nonumber
& =I(\mathbf{x};\mathbf{y}_1|v_1) +  I(v_1;\mathbf{y}_2|v_2)+ \ldots+I(v_{K-2};\mathbf{y}_{K-1}|v_{K-1})+ \sum_{r=K}^{M+K-1} \mathds{h}(y_r)- \mathds{h}(y_r|v_{K-1})\\ \nonumber
& \leq I(\mathbf{x};\mathbf{y}_1|v_1) +  I(v_1;\mathbf{y}_2|v_2)+ \ldots+I(v_{K-1};\mathbf{y}_{K-1}|v_{K-1}) \\
 \nonumber &
+\sum_{r=K}^{M+K-1} \mathds{h}(\mathbf{h}_r^{\dagger}\mathbf{x}+z_r)- \sum_{r=K}^{M+K-1} \mathds{h}(\mathbf{h}_r^{\dagger}\mathbf{x}+z_r|v_{K-1})
\\ \nonumber
& \overset{(a)}{\leq} I(\mathbf{x};\mathbf{y}_1|v_1) +  I(v_1;\mathbf{y}_2|v_2)+ \ldots+I(v_{K-2};\mathbf{y}_{K-1}|v_{K-1})+ \sum_{r=K}^{M+K-1} \mathds{h}(\mathbf{h}_r^{\dagger}\mathbf{x}+z_r)-  \mathds{h}(\mathbf{y}_K|v_{K-1})\\ \nonumber
& =  \mathds{h}(\mathbf{y}_1|v_1)-\mathds{h}(\mathbf{y}_1|v_1,\mathbf{x})+\mathds{h}(\mathbf{y}_2|v_2)-\mathds{h}(\mathbf{y}_2|v_2,v_1)+\ldots+\mathds{h}(\mathbf{y}_{K-1}|v_{K-1})-\mathds{h}(\mathbf{y}_{K-1}|v_{K-2},v_{K-1})\\ \nonumber
&- \mathds{h}(\mathbf{y}_K|v_{K-1}) + \sum_{r=K}^{M+K-1} \mathds{h}(\mathbf{h}_r^{\dagger}\mathbf{x}+z_r) \\ \nonumber
& \overset{(b)}{=}  \mathds{h}(\mathbf{y}_1|v_1)-\mathds{h}(\mathbf{y}_1|v_1,\mathbf{x})+\mathds{h}(\mathbf{y}_2|v_2)-\mathds{h}(\mathbf{y}_2|v_1)+\ldots+\mathds{h}(\mathbf{y}_{K-1}|v_{K-1})-\mathds{h}(\mathbf{y}_{K-1}|v_{K-2})\\ \nonumber
&- \mathds{h}(\mathbf{y}_K|v_{K-1}) + \sum_{r=K}^{M+K-1} \mathds{h}(\mathbf{h}_r^{\dagger}\mathbf{x}+z_r),
\end{align}
where (a) relies on the fact that $\sum_{r=K}^{M+K-1} \mathds{h}(\mathbf{h}_r^{\dagger}\mathbf{x}+z_r|v_{K-1}) \geq  \mathds{h}(\mathbf{h}_{K}^{\dagger}\mathbf{x}+z_{K},\ldots, \mathbf{h}_{M+K-1}^{\dagger}\mathbf{x}+z_{M+K-1} |v_{K-1})=\mathds{h}(\mathbf{y}_K|v_{K-1})$ and (b) relies on the Markov chain $v_{K-1} \ldots \leftrightarrow v_{1}   \leftrightarrow \mathbf{x}  \leftrightarrow (\mathbf{y}_1,\ldots, \mathbf{y}_K, y_{K+1},\ldots, y_{M+K-1})$.

Note that for $r < K$
\begin{align}\label{eq:2}
\mathds{h}(\mathbf{y}_r|v_r)-\mathds{h}(\mathbf{y}_{r+1}|v_r) & \\
& =  \mathds{h}(y_{r},  \mathbf{y}_{r+1}|v_r)-\mathds{h}(\mathbf{y}_{r+1}|v_r)
\\ &= \mathds{h}(y_{r} |  \mathbf{y}_{r+1}, v_r) \\
& \overset{(a)}{\leq} \mathds{h}(y_{r} |  \mathbf{y}_{K}, v_r) \\
&  \overset{(b)}{=} \mathds{h}(\boldsymbol{\phi}_{r}^{\dagger}\mathbf{y}_{K}^{\dagger}-\boldsymbol{\phi}_{r}^{\dagger}\mathbf{z}_{K}+z_r | \mathbf{y}_{K},v_r) \\
& = \mathds{h}(-\boldsymbol{\phi}_{r}^{\dagger}\mathbf{z}_{K}+z_r |  \mathbf{y}_{K}, v_r) \\
& \leq \mathds{h}(-\boldsymbol{\phi}_{r}^{\dagger}\mathbf{z}_{K}+z_r)
\end{align}
where (a) relies on the fact that $\mathbf{y}_{r+1}=[y_{r+1},\ldots,y_{K+1}, y_{K},\ldots,y_{M+K-1}]^{\dagger}$. In addition (b) relies on the  fact that $\mathbf{H}_K$ is full-rank and therefore $\mathbf{h}_r \in \textrm{Span}\{\mathbf{h}_{K}, \mathbf{h}_{K+1}, \ldots, \mathbf{h}_{M+K-1}\}$, for $r=1, \ldots, K-1$. Therefore, $\mathbf{h}_r =  \mathbf{H}_K \boldsymbol{\phi}_{r}$ for a vector $\boldsymbol{\phi}_{r} \in \mathds{R}^{M \times 1}$.
Then, $y_r=\mathbf{h}_r^{\dagger}\mathbf{x}+z_r=\boldsymbol{\phi}_{r}^{\dagger}\mathbf{H}_K^{\dagger}\mathbf{x}+z_r=\boldsymbol{\phi}_{r}^{\dagger}\mathbf{y}_{K}^{\dagger}-\boldsymbol{\phi}_{r}^{\dagger}\mathbf{z}_{K}+z_r$.

In addition, note that
\begin{eqnarray}\label{eq:3}
\mathds{h}(\mathbf{y}_1|v_1,\mathbf{x})=\mathds{h}(\mathbf{z}_1)=\sum_{i=1}^{M+K-1}\mathds{h}(z_i).
\end{eqnarray}
Therefore, using~\eqref{eq:1}, \eqref{eq:2}, and \eqref{eq:3}, we have,
\begin{align}
&R_1+\ldots+R_{K-1}+MR_K  \\
& \leq  \sum_{r=1}^{K-1} \mathds{h}(-\boldsymbol{\phi}_{r}^{\dagger}\mathbf{z}_{K}+z_r)+  \sum_{r=K}^{M+K-1} \mathds{h}(\mathbf{h}_r^{\dagger}\mathbf{x}+z_r)-\sum_{r=1}^{M+K-1}\mathds{h}(z_{r}) \\
& =\sum_{r=K}^{M+K-1} I(\mathbf{h}_r^{\dagger}\mathbf{x}+z_r, \mathbf{x}) + \sum_{r=1}^{K} \mathds{h}(-\boldsymbol{\phi}_{r}^{\dagger}\mathbf{z}_{K}+z_r)-\sum_{r=1}^{K-1}\mathds{h}(z_{r})
\end{align}
Note that the RHS of the above equation is in the order of $\frac{M}{2}\log_2 P$, and therefore the result is concluded.
\end{proof}

Using the above theorem, the converses of Theorems~\ref{thm:ach2} and ~\ref{thm:ach3} are easily proven.

\textbf{Converse for Theorem~\ref{thm:ach2}:}
\begin{proof}
From Theorem~\ref{thm:main}, we have
\begin{eqnarray}
\lim_{P \rightarrow \infty}\frac{R_1+\ldots+R_{r-1}+MR_r+R_r+\ldots+R_K}{0.5\log_2 P} \leq M, \ r=1,\ldots,K.
\end{eqnarray}
Adding the above $K$ inequalities, we have
\begin{eqnarray}
(M+K-1)\lim_{P \rightarrow \infty}\frac{R_1+\ldots+R_K}{0.5\log_2 P} \leq MK, \ r=1,\ldots,K.
\end{eqnarray}
\end{proof}

\textbf{Converse for Theorem~\ref{thm:ach3}:}
\begin{proof}
It is easily followed from Theorem~\ref{thm:main} and the fact that $\lim_{P \rightarrow \infty}\frac{R_r}{0.5\log2 P} \leq 1$ for $r=1,\ldots,M-1$.
\end{proof}

\section{Complex Channels}\label{sec:complex}
In~\cite{abolfazl_k_2}, a generalized version of Khintchin-Groshev Theorem (see~\cite{groshev-berink, groshev-Beresnevich}) has been used to establish a pseudo-MIMO approach for the interference management in real channels. For complex channels, a traditional approach is to transform the channel to a real channel with twice dimensions. However, in the transformed channel, the channel coefficients are not independent any more, and therefore, using the result of ~\cite{abolfazl_k_2} over the transformed channel is not straight-forward. Here in this section, we borrow a recent result from Number Theory to directly extend the machinery of~\cite{abolfazl_k_2} to complex channels. Therefore,  we can simply extend the results that we derived in the previous sections to complex channels as well.

Let us use the notation $\mathds{K}$ for $\mathds{C}$ or $\mathds{R}$. For any vector $\boldsymbol{\nu} \in \mathds{K}^{{\xi}-1}$, the multiplicative Diophantine exponent $\omega(\boldsymbol{\nu})$ is defined as

\begin{align}
\omega( \boldsymbol{\nu} ) = \sup \left\{ \eta  \Big{|} \
 |\sum_{i=1}^{{\xi}-1} \nu_i q_i+ p|  \leq (\frac{1}{\prod_{i=1}^{\xi-1} \max(1,q_i)})^{{\frac{\eta}{\xi-1}}} \ \textrm{for infinitely many}\ (p, q_1,q_2,...,q_{{\xi}-1})\in \mathds{Z}^{{\xi}} \right \}.
\end{align}

\begin{thm}\label{thm:comp}
Consider the mapping $\boldsymbol{\psi}=(\psi_1,\psi_2,\ldots, \psi_{{\xi}-1})$ from an open subset $\mathcal{U} \subset \mathds{K}^{d}$ to   $\mathds{K}^{{\xi}-1}$.
If  $1,\psi_1,\ldots, \psi_{{\xi}-1}$ are linearly independent in $\mathds{R}$, then
\begin{itemize}
\item For $\mathds{K}=\mathds{R}$,  $\omega( \boldsymbol{\psi} (\mathbf{x} ))$ is equal to ${\xi}-1$ for almost all $\mathbf{x} \in \mathcal{U}$~\cite{groshev-berink, groshev-Beresnevich}.
\item For $\mathds{K}=\mathds{C}$,  $\omega( \boldsymbol{\psi} (\mathbf{x} ))$ is equal to $\frac{{\xi}-2}{2}$ for almost all $\mathbf{x} \in \mathcal{U}$~\cite{complex_kh}.
\end{itemize}

\end{thm}

In~\cite{abolfazl_k_2}, the case where~$\mathds{K}=\mathds{R}$ has been addressed. Here,  we focus on the case, where $\mathds{K}=\mathds{C}$.

Consider the  mappings $\boldsymbol{\psi} =(\psi_1,\psi_2,\ldots, \psi_{{\xi}-1})$ from an open subset $\mathcal{U} \subset \mathds{C}^{d}$ to   $\mathds{C}^{{\xi}-1}$, where
$1,\psi_1,\ldots, \psi_{{\xi}-1}$ are linearly independent in $\mathds{R}$. Then, Theorem~\ref{thm:comp} states that for $\eta =\frac{{\xi}-2}{2}+\epsilon$, $\epsilon >0$,  for almost all $\mathbf{x}\in \mathcal{U}$ and $(q_1,\ldots,q_{{\xi}-1},p)\in \mathds{Z}^{{\xi}}$, we have,
\begin{eqnarray}
|p+ \sum_{i=1}^{{\xi}-1} \psi_i(\mathbf{x})q_i| > \left(\frac{1}{\prod_{i=1}^{\xi-1}\max(1,q_i)}\right)^{\frac{\eta}{\xi-1}}.
\end{eqnarray}

Consider the Gaussian multiple-access channel with ${\xi}$ inputs with the channel gains $\nu_i$, $i=1,\ldots,{\xi}$, modeled as,
\begin{align}
y=\sum_{i=1}^{{\xi}}\nu_i x_i+ z,
\end{align}
where $\nu_{{\xi}}=1$,  $\nu_i=\psi_i(\mathbf{x})$ for an $\mathbf{x} \in \mathcal{U}$,  $i=1,\ldots,{\xi}-1$. Moreover, $z$ denotes  Gaussian complex noise with $z\sim\mathcal{CN}(0,1)$. Let us use the following constellation for each input,
\begin{eqnarray}
\{\lambda u | u \in \mathds{Z}, -Q< u <Q \}.
\end{eqnarray}
We choose $Q$ as
\begin{align}
Q=&\gamma P^{ \frac{1-\epsilon}{{\xi}+2\epsilon}},
\end{align}
where $\gamma$ is a constant. It is important to note that the rate of the constellation is at least
\begin{align}
\log_2{2Q-1} \cong   \frac{1-\epsilon}{{\xi}+2\epsilon} \log_2 (P) + \log (2\gamma_1).
\end{align}
This means that the rate of the constellation at each transmitter is about $\frac{1}{{\xi}} \log_2 (P)$. This is almost twice of the rate of the constellations we use at the transmitters for the real channels (see~\cite{abolfazl_k_2}).

We choose $\lambda$ such that $\mathds{E}[x_i^2] \leq \gamma_2 P$, for a constant $\gamma_2$. We have
$\mathds{E}[x_i^2] \leq \lambda^2 Q^2$. Therefore,
\begin{align}
\lambda = & \gamma_3 P^{ \frac{{\xi}-2+4\epsilon}{2({\xi}+2\epsilon)}},
\end{align}
where $\gamma_3$ is a function of $\gamma$, $\gamma_1$, ${\xi}$, and $\epsilon$.

Then the received symbol $y$ is a noisy version of a point from the following constellation,
\begin{align}
\mathcal{C}=\{ \lambda (\sum_{i=1}^{{\xi}}\nu_i u_i+ z), \textrm{where} \  u_i \in \mathds{Z} \cap (-Q, Q) \}.
\end{align}
From Theorem~\ref{thm:comp}, we know that the minimum distance of the constellation $\mathcal{C}$ is
\begin{align}
d_{\min}=\frac{\lambda}{(\max_i{q_i})^{\eta}}=\frac{\lambda}{Q^{\frac{{\xi}-2}{2}+\epsilon}},
\end{align}
almost surely. Then, it is easy to see that
\begin{align}
d_{\min}=\gamma_4 P^{\frac{\epsilon}{2}}.
\end{align}
%
%

We have the following observations:
\begin{itemize}
\item[(i)] Since the minimum distance of the constellation is not zero almost surely, then there is a one-to-one mapping between the points of the constellation $\mathcal{C}$ and vectors $(u_1,\ldots, u_{{\xi}})$.
\item[(ii)] Since the minimum distance of the received constellation is growing with $P^{\frac{\epsilon}{2}}$, therefore $P_e$, the probability of incorrectly detecting a point of  the constellation from $y$, goes to zero, as  $P$ grows.
\end{itemize}
Then, using Fano's inequality, we can easily show that the rate of each transmitter is
$\log_2{2Q} \cong   \frac{1-2\epsilon}{{\xi}+2\epsilon} \log_2 (P) + \log (2\gamma)$. This means that each transmitter achieves the \emph{complex} DoF of $\frac{1}{{\xi}}$.

Using this approach, we can extend all the results presented in this paper to complex channels. In fact, in schemes presented in~\ref{sec:scheme1} and~\ref{sec:scheme2}, we choose $Q$ as $Q=\gamma P^{ \frac{1-\epsilon}{{\xi}+2\epsilon}}$. Then, we can achieve the rate twice of what is achievable for real channels. Since the DoF of the complex channel is defined as $d= \lim_{P \rightarrow \infty} \frac{C_{sum}}{\log_2(P)}$, then we achieve the same DoF of the real channels. 

\textbf{Remark:} Consider a compound complex $X$ channel with $M$ transmitters and $K$ receivers. Therefore, as proved, the complex DoF of this channel is equal to $\frac{MK}{M+K-1}$. Now let us transform the channel to a real channel with $2M$ transmitters and $2K$ receivers.  If we ignore the possibility of cooperation between real and imaginary parts of each transmitter/receiver, then we will have a real $X$ channel with $2M$ transmitters and $2K$ receivers. Therefore, the real DoF of the resulting channel is upper-bounded by $\frac{4MK}{2M+2K-1}$. This means that the complex DoF of the resulting channel is upper-bounded by $\frac{2MK}{2M+2K-1}$ or $\frac{MK}{M+K-0.5}$. We note that $\frac{MK}{M+K-0.5} \leq \frac{MK}{M+K-1}$. This means that ignoring the possibility of cooperation among the real and imaginary components of each transmitter/receiver results in a sub-optimal scheme.
\appendices

\section{Achievable Scheme For Theorem~\ref{thm:ach1}}\label{app:1}
In Section~\ref{sec:scheme1}, we explained the achievable scheme for Theorem~\ref{thm:ach1} for $M=2$ and $K=2$. Here, we explain it  step-by-step  for general value of $M$ and $K$.

\subsection{Encoding}
  Assume that the base station has message $W_r$ for receiver $r$ . $W_r$ is divided  into $M$ independent parts, i.e. $W_r=(W_{r1},W_{r2},\ldots, W_{rM})$. $W_{rt}$ will be sent through transmitter $t$ for receiver $r$. $W_{rt}$ itself is divided into $L_r$ parts, $W_{rt}=(W_{rt}^{(1)}, W_{rt}^{(2)},\ldots, W_{rt}^{(L_r)} )$. Then,
  $W^{(l)}_{rt}$ is encoded into the sequence $(u_{rt}^{(l)}[1], u_{rt}^{(l)}[2], \ldots, u_{rt}^{(l)}[T])$, where $T$ is the length of the codeword, and $u_{rt}^{(l)}[m]$, $m=1, \ldots, T$,  belongs to the integer constellation $(-Q, Q)_{\mathds{Z}}$.  The sequence $(u_{rt}^{(l)}[1], u_{rt}^{(l)}[2], \ldots, u_{rt}^{(l)}[T])$ is weighted by the modulation pseudo-vector $\nu^{(l)}_{rt}$. Each transmitter sends a weighted linear combination of the corresponding codewords. More precisely,
\begin{eqnarray}
x_t[m]= \lambda  \sum_{r=1}^{K} \sum_{l=1}^{L_{r}} \nu^{(l)}_{rt} u^{(l)}_{rt}[m].
\end{eqnarray}
 The normalizing constant $\lambda$ is used to guarantee the power constraint.

\subsection{Modulation Pseudo-Vectors}
Let us define the set $\mathcal{B}_r$, for $r=1,\ldots, K$, as follows:
\begin{eqnarray}
\mathcal{B}_{r}=\left \{ \prod_{r'=1, r' \neq r}^{K} \prod_{s=1}^{J_{r'}} \prod_{t=1}^{M} (h_{r't}^{\{s\}} )^{ \alpha_{r't}^{\{s\}}  }, \quad   1 \leq \alpha_{r't}^{\{s\}} \leq n_r, \quad r'\neq r   \right \},
\end{eqnarray}
where $n_r$ is a constant number. We use $\mathcal{B}_{r}$  as the set of the modulation pseudo-vectors for data sub-streams intended for receiver $r$. Therefore,
\begin{eqnarray}
\left \{ \nu^{(l)}_{rt}, \quad l=1,\ldots, L_{r} \right \}  = \mathcal{B}_{r}.
\end{eqnarray}
Note that $\nu^{(l)}_{rt} =\nu^{(l')}_{rt}$, if $l \neq l'$. Therefore,  $L_r$ is equal to $|\mathcal{B}_r|$. It is easy to see that $L_r=n_r^{M(\sum_{\hat{r}=1}^{K}J_{\hat{r}}-J_r)}$. Consider $L$ as a large enough integer. We  choose $n_r$
\begin{eqnarray}
n_r= \lfloor  L^{\frac{1}{ M(\sum_{\hat{r}=1}^{K}J_{\hat{r}}-J_r) }} \rfloor.
\end{eqnarray}
Therefore,  we have $L_r=L+o(L)$, for $r=1,\ldots,K$.

\subsection{Received Signals and Interference Alignment}
Let us focus on receiver $r$, when the channel state is $\mathbf{h}_r^{\{\hat{s}\}}$,  $\hat{s} \in \{ 1,\ldots, J_r\}$. The received signal is given by,
\begin{align}
y_r^{\{ \hat{s} \}}[m]& = \mathbf{h}_r^{{\{ \hat{s} \}}^{\dagger}} \mathbf{x}[m]+z_r^{\{ \hat{s}\}}[m] =\\
& =  \sum_{t=1}^{M}h_{rt}^{\{ \hat{s} \}} x_t[m]+z_1^{\{ \hat{s} \}}[m] = \\
& =  \lambda \sum_{t=1}^{M}h_{rt}^{\{ \hat{s} \}} \sum_{\hat{r}=1}^{K} \sum_{l=1}^{L_{\hat{r}}} \nu^{(l)}_{\hat{r}t} u^{(l)}_{\hat{r}t}[m]
+z_r^{\{ \hat{s}\}}[m] \\ \label{eq:4}
&
=  \lambda   \sum_{t=1}^{M} \sum_{l=1}^{L_{r}} h_{rt}^{\{ \hat{s} \}} \nu^{(l)}_{rt} u^{(l)}_{rt}[m] +
 \lambda \sum_{\hat{r}=1, \hat{r} \neq r}^{K}  \sum_{t=1}^{M} \sum_{l=1}^{L_{\hat{r}}} h_{rt}^{\{ \hat{s} \}} \nu^{(l)}_{\hat{r}t} u^{(l)}_{\hat{r}t}[m] +z_r^{\{ \hat{s}\}}[m].
\end{align}
Note that the first summation in the RHS of the above equations conveys information for receiver $r$, while the last summation is just interference for this receiver.

Since $\nu^{(l)}_{rt} \in \mathcal{B}_r$, then for any $r$ and $\hat{r}$, we have
\begin{eqnarray}
h_{rt}^{\{ \hat{s} \}} \nu^{(l)}_{\hat{r}t} \in h_{rt}^{\{ \hat{s} \}}. \mathcal{B}_{\hat{r}}.
\end{eqnarray}

We  observe three important properties:
\begin{itemize}
\item[(1)] If $t \neq \hat{t}$, then $h_{rt}^{\{ \hat{s} \}} \neq h_{r\hat{t}}^{\{ \hat{s} \}}$, almost surely. Therefore,
\begin{eqnarray}
 h_{rt}^{\{ \hat{s} \}}. \mathcal{B}_r \cap h_{r\hat{t}}^{\{ \hat{s} \}}. \mathcal{B}_r =\emptyset, \quad \forall \ t, \hat{t} \in \{1,\ldots,M \}, \ t \neq \hat{t}
\end{eqnarray}
Therefore, $ \bigcup_{t=1}^{M} h_{rt}^{\{ \hat{s} \}}. \mathcal{B}_r  =ML_r$, almost surely.  This means that at receiver $r$, $ML_r$ favorite data sub-streams are received with distinct coefficients.

\item[(2)]
It is easy to see that
\begin{eqnarray}
\left( \bigcup_{t=1}^{M} h_{rt}^{\{ \hat{s} \}}. \mathcal{B}_r \right) \bigcap \left(  \bigcup_{t=1}^{M} h_{rt}^{\{ \hat{s} \}}. \mathcal{B}_{\hat{r}} \right) = \emptyset.\ \forall \hat{r}, \ \hat{r}\neq r
\end{eqnarray}
This means that interference sub-streams are received at receiver $r$ with coefficients which are different from the coefficients of the favorite data sub-streams.

\item[(3)] Now, in~\eqref{eq:4}, we focus on the coefficients of the data sub-streams, intended for Receiver $\hat{r}$, $\hat{r}\neq r$, and cause interference at receiver $r$. More precisely, we focus on the coefficients of
     $\sum_{t=1}^{M} \sum_{l=1}^{L_{\hat{r}}} h_{rt}^{\{ \hat{s} \}} \nu^{(l)}_{\hat{r}t} u^{(l)}_{\hat{r}t}[m]$. Apparently, the coefficients $h_{rt}^{\{ \hat{s} \}} \nu^{(l)}_{\hat{r}t}$ are belong to $\bigcup_{t=1}^{M} h_{rt}^{\{ \hat{s} \}}. \mathcal{B}_{\hat{r}}$. However, it is easy to see that,
\begin{eqnarray}
| \bigcup_{t=1}^{M} h_{rt}^{\{ \hat{s} \}}. \mathcal{B}_{\hat{r}} | =n_{\hat{r}}^{M(\sum_{\bar{r}=1}^{K}{J_{\bar{r}}}-J_{\hat{r}}-1)}(n_{\hat{r}}+1)^M, \quad \hat{r} \neq {r}.
\end{eqnarray}

Remember that $| h_{rt}^{\{ \hat{s} \}}. \mathcal{B}_{\hat{r}}| =(n_{\hat{r}})^{M(\sum_{\bar{r}=1}^{K}{J_{\bar{r}}}-J_{\hat{r}})}$. This means that $| \bigcup_{t=1}^{M} h_{rt}^{\{ \hat{s} \}}. \mathcal{B}_{\hat{r}} |$ has almost the same cardinality as $|  h_{rt}^{\{ \hat{s} \}}. \mathcal{B}_{\hat{r}}|$, for $\hat{r} \neq r$  and $t=1, \ldots,M$.  It
 implies that the sets $h_{rt}^{\{ \hat{s} \}}. \mathcal{B}_{\hat{r}}$, $t=1,\ldots,M$, are almost the same with just a few different elements (compared to the size of each set). The interference sub-streams which arrived with the same coefficients are in fact aligned.
\end{itemize}

At receiver $r$, we merge the interference sub-streams with the similar coefficients, so we have
\begin{align}\label{eq:5}
y_r^{\{ \hat{s} \}}[m]=  \lambda   \sum_{t=1}^{M} \sum_{l=1}^{L_{r}} h_{rt}^{\{ \hat{s} \}} \nu^{(l)}_{rt} u^{(l)}_{rt}[m] +
 \lambda \sum_{\hat{r}=1, \hat{r} \neq r}^{K}  \sum_{l=1}^{\kappa_{\hat{r}}}  \bar{\nu}^{(l)}_{\hat{r},r,\hat{s}} \bar{u}^{(l)}_{\hat{r},r,\hat{s}}[m] +n_r^{\{ \hat{s}\}}[m],
\end{align}
where
$\kappa_{\hat{r}} = n_{\hat{r}}^{M(\sum_{\bar{r}=1}^{K}{J_{\bar{r}}}-J_{\hat{r}}-1)}(n_{\hat{r}}+1)^M$, and $\bar{\nu}^{(l)}_{\hat{r},r,\hat{s}} \in \bigcup_{t=1}^{M} h_{rt}^{\{ \hat{s} \}}. \mathcal{B}_{\hat{r}}$. In addition, $ \bar{u}^{(l)}_{\hat{r},r,\hat{s}}[m] \in (-MQ,MQ)_{\mathds{Z}}$. Therefore, we have noisy version of the integer combination of $\sum_{\hat{r}=1 \ \hat{r} \neq r }^{K}\kappa_{\hat{r}}+ML_r$ real numbers.
These real numbers have another important property. All of these numbers are monomial functions of channel coefficients and  these monomial functions are linearly independent.

Note that $\frac{ML_r}{\sum_{\hat{r}=1, \ \hat{r} \neq r }^{K}\kappa_{\hat{r}}+ML_r}$ the sub-streams in~\eqref{eq:5} carries favorite message for receiver $r$.  Since $\kappa_{\hat{r}}= L+o(L)$ and $L_r=L+o(L)$, then $\frac{ML_r}{\sum_{\hat{r}=1, \ \hat{r} \neq r }^{K}\kappa_{\hat{r}}+ML_r} \simeq \frac{M}{M+K-1}$.

Note that at each receiver, the total available DoF is just one. Here we develop a signaling scheme such that each data sub-stream has DoF of $\frac{1}{\xi}$ DoF, where
\begin{eqnarray}
\xi= \max_{r}\{\sum_{\hat{r}=1}^{K} \kappa_{\hat{r}} -\kappa_{r} +ML_r \}.
\end{eqnarray}
Therefore, at receiver $r$, a $\frac{ML_r}{\xi} \simeq \frac{M}{M+K-1}$ portion of the available DoF is used for receiving favorite data sub-streams, while   $\frac{\sum_{\hat{r}=1}^{K} \kappa_{\hat{r}} -\kappa_{r}}{\xi}\simeq 1- \frac{M}{M+K-1}$ is wasted for interference.

\subsection{Choosing $Q$ and $\lambda$}
Now we choose $Q$ as follows:
\begin{eqnarray}
Q = (\frac{P}{M})^{\frac{1-\epsilon}{2(\xi + \epsilon) }},
\end{eqnarray}
where $\epsilon$ is an arbitrary small constant. Note that $u^{(l)}_{rt}[m]$ is from the integer constellation in  $(-Q,Q)$, where the rate of this constellation is $\log_2 (2Q)= \frac{1-\epsilon}{2(\xi + \epsilon)} \log_2(\frac{P}{M})+ 1$. It is easy to see that,
\begin{align}
\mathbb{E} [x^2_t[m]]&=\lambda^2 \Gamma^2 Q^2,
\end{align}
where
\begin{eqnarray}
\Gamma^2= \sum_{r=1}^{K} \sum_{l=1}^{L_{r}} \left(\nu^{(l)}_{rt}\right)^2=\sum_{r=1}^K \sum_{\nu \in \mathcal{B}_r} \nu_r^2.
\end{eqnarray}
We choose $\lambda$ such that
\begin{eqnarray}
\mathbb{E} [x^2_t[m]] \leq \frac{P}{M}.
\end{eqnarray}
One choice for $\lambda$ is
\begin{eqnarray}
\lambda =\frac{P^{\frac{1}{2}}}{\sqrt{M}\Gamma Q}=\frac{1}{\Gamma} \left(\frac{P}{M}\right)^{ \frac{\xi-1+2\epsilon}{2(\xi+\epsilon)}}.
\end{eqnarray}

\subsection{Constellation Formed At Each Receiver}
At receiver $r$, when the channel state is $\mathbf{h}_r^{\{\hat{s}\}}$,  $\hat{s} \in \{ 1,\ldots, J_r\}$,  the received signal at  time $m$ is a noisy version of a point from the constellation $\mathcal{C}_r$, where
\begin{eqnarray*}
\mathcal{C}_r = \left\{  \lambda    \sum_{t=1}^{M} \sum_{l=1}^{L_{r}} h_{rt}^{\{ \hat{s} \}} \nu^{(l)}_{rt} u^{(l)}_{rt} +
 \lambda \sum_{\hat{r}=1, \hat{r} \neq r}^{K}  \sum_{l=1}^{\kappa_{\hat{r}}}  \bar{\nu}^{(l)}_{\hat{r},r,\hat{s}} \bar{u}^{(l)}_{\hat{r},r,\hat{s}} , u^{(l)}_{21} \in (-Q,Q)_{\mathds{Z}}, \ \bar{u}^{(l)}_{\hat{r},r,\hat{s}} \in (-MQ,MQ)_{\mathds{Z}}
\right\}.
\end{eqnarray*}
Using the Theorem 4 of~\cite{abolfazl_k_2}, we can show that the minimum distance of the constellation is $\eta \left( \frac{P}{M} \right)^{\epsilon}$, almost surely. Here $\eta= \frac{1}{ \Gamma  M^{ (\xi+\epsilon) } }$.

 This means that
\begin{itemize}
\item[(i)] There is a one to one mapping between $\left(\left\{ u^{(l)}_{rt} \right\}_{\begin{smallmatrix} t=1,\ldots,M \\ l=1,\ldots,L_r \end{smallmatrix}},  \left\{\bar{u}^{(l)}_{\hat{r},r,\hat{s}}\right\}_{\begin{smallmatrix} \hat{r}=1,\ldots,K, \  \hat{r} \neq r \\ l=1,\ldots,\kappa_{\hat{r}} \end{smallmatrix} } \right)$ and points of the constellation $\mathcal{C}_r$.
\item[(ii)] In high power, we can de-noise the received signal, the detect the point of the constellation $\mathcal{C}_r$ with vanishing probability of error and find the unique corresponding  \\
    $\left(\left\{ u^{(l)}_{rt}[m] \right\}_{\begin{smallmatrix} t=1,\ldots,M \\ l=1,\ldots,L_r \end{smallmatrix}},  \left\{\bar{u}^{(l)}_{\hat{r},r,\hat{s}}[m]\right\}_{\begin{smallmatrix} \hat{r}=1,\ldots,K, \  \hat{r} \neq r \\ l=1,\ldots,\kappa_{\hat{r}} \end{smallmatrix} } \right)$.
\end{itemize}

 Then, we pass the string $(\hat{u}^{(l)}_{rt}[1],\hat{u}^{(l)}_{rt}[2],\ldots, \hat{u}^{(l)}_{rt}[T])$ to the decoder to decode $\hat{W}_{rt}^{(l)}$, for $l=1,\ldots, L_r$, and $t=1,\ldots, M$.

\subsection{Performance Analysis}
 Probability of error of detecting  $\left(\left\{ \hat{u}^{(l)}_{rt}[m] \right\}_{\begin{smallmatrix} t=1,\ldots,M \\ l=1,\ldots,L_r \end{smallmatrix}},  \left\{\hat{\bar{u}}^{(l)}_{\hat{r},r,\hat{s}}[m]\right\}_{\begin{smallmatrix} \hat{r}=1,\ldots,K, \  \hat{r} \neq r \\ l=1,\ldots,\kappa_{\hat{r}} \end{smallmatrix} } \right)$ goes to zero as $P \rightarrow \infty$.
 Then, using the fact that $u^{(l)}_{rt}[m]$ is from the integer constellation in  $(-Q,Q)$ with rate $\log_2 (2Q)= \frac{1-\epsilon}{2(\xi + \epsilon)} \log_2(\frac{P}{M})+ 1$, we can show that each of the data sub-streams $W_{rt}^{(l)}$
 achieves the DoF of $\frac{1-\epsilon}{\xi + \epsilon}$. Therefore, we achieve the DoF of
$ML_r \frac{1-\epsilon}{\xi + \epsilon}$ at receiver $r$. Thus, we achieve the total DoF of
\begin{eqnarray}
(M\sum_{r=1}^{K}L_r)\frac{1-\epsilon}{\xi + \epsilon}.
\end{eqnarray}
Since $L_r=L+o(L)$, and $\xi =(K-1)L+ML+o(L)$,  by choosing large enough $L$ and small enough $\epsilon$, we can achieve a DoF, arbitrary close to $\frac{MK}{M+K-1}$.

\section*{Acknowledgment}
{The author would like to thank Professor David Tse for many helpful comments and long discussions.

\bibliographystyle{IEEE}

\end{document}